\documentclass[]{article}

\usepackage{array}
\usepackage{amssymb,amsmath}		
\usepackage{graphics}		
\usepackage{longtable}          
\usepackage{subfigure}
\usepackage{float}

\usepackage{anysize}
\marginsize{1.5cm}{1.5cm}{2cm}{2cm}

\newcommand{\vectornorm}[1]{{\left|\left|#1\right|\right|}_{\ell_2}} 
 
\newcommand{\vectornormInf}[1]{{\left|\left|#1\right|\right|}_{\ell_\infty}} 
\newcommand{\vectornormZero}[1]{{\left|\left|#1\right|\right|}_{\ell_0}} 
\newtheorem{theorem}{Theorem}[section]

\newtheorem{corollary}[theorem]{Corollary}
\newenvironment{proof}[1][Proof]{\begin{trivlist}
\item[\hskip \labelsep {\bfseries #1}]}{\end{trivlist}}

\title{Non-Adaptive Distributed Compression in Networks}
\author{Mahdy Nabaee and Fabrice Labeau
\thanks{M. Nabaee and F. Labeau are with the Department
of Electrical and Computer Engineering, McGill University, Montreal.}
\thanks{This work was supported by Hydro-Québec, the Natural Sciences and Engineering Research Council of Canada and McGill University in the framework of the NSERC/Hydro-Québec/McGill Industrial Research Chair in Interactive Information Infrastructure for the Power Grid.}
}

\begin{document}

\maketitle

\begin{abstract}
In this paper, we discuss non-adaptive distributed compression of inter-node correlated real-valued messages.
To do so, we discuss the performance of conventional packet forwarding via routing, in terms of the total network load versus the resulting quality of service (distortion level).
As a better alternative for packet forwarding, we briefly describe our previously proposed one-step Quantized Network Coding (QNC), and make motivating arguments on its advantage when the appropriate marginal rates for distributed source coding are not available at the encoder source nodes.
We also derive analytic guarantees on the resulting distortion of our one-step QNC scenario.
Finally, we conclude the paper by providing a mathematical comparison between the total network loads of one-step QNC and conventional packet forwarding, showing a significant reduction in the case of one-step QNC.
\end{abstract}


\section{Introduction}
\label{sec:Intro}

In this paper, we study the incast of correlated real-valued messages (\textit{e.g.} sensed data) in multi-hop networks.
Usually, the messages in such networks need to be transmitted to a so called gateway node, which relays them to the next level of the hierarchy.
Delivery of these messages to the gateway node is based on the hopping nature of the network, in which the messages are forwarded, with the aid of a transport protocol.
Specifically, in this paper, we review the information theoretic results on this so called data gathering (or incast) scenario.
We will also motivate a discussion on non-adaptive encoding of messages and describe our previously proposed transmission scheme, which tries to perform as efficient as when the appropriate marginal rates for separate encoding were known.
Our notion of non-adaptive encoding is with respect to the inter-node correlation of messages, and studies the cases, where the knowledge of appropriate marginal rates are not available at the encoders' side.

Initiated in \cite{netInfFlow}, network information flow analysis has helped develop more efficient transmission schemes, when dealing with multiple communication agents. 
Explicitly, network coding has been suggested as an alternative for routing-based packet forwarding, for which many theoretical and practical advantages have been listed.
In the case of lossless networks, with limited link capacities, cut-set upper bound has been shown to be the capacity region \cite{netInfFlow}.
As an extension to lossy networks, in \cite{erasNetCap}, the capacity region is calculated for erasure networks, with side information at the decoder.
Moreover, random linear network coding \cite{koetter2003algebraic}, which can be implemented non-adaptively and computationally simple, is shown to be sufficient for multicast of messages in memoryless and lossless networks \cite{NC_RLNCtoMulticast}.

For the case of correlated sources (messages), performing distributed lossy or lossless source coding \cite{xiong2004distributed} has shown to be sufficient, even if adopted along routing-based packet forwarding \cite{SWCtheorem,ho2004network,NCCorr_NIFwithCOrrSo}.
However, advantages of network coding have drawn attention, especially in the case of lossy networks \cite{fragouli2009network}.
As a new advantage of network coding, in cases where the knowledge of appropriate information rates for performing distributed source coding is not available at the encoders' side, random linear network coding seems to be helpful.
Specifically, this has been studied in \cite{naba1,naba3}, where we proposed Quantized Network Coding (QNC) and discussed practical feasibility of distributed source coding by semi-random linear network coding.

In this paper, we study non-adaptive distributed compression of correlated messages (sources), by using network coding. 
In section~\ref{sec:ProbDesc}, we describe the network and define the incast scenario for sparse messages.
Then, in section~\ref{sec:Qforward}, we briefly describe quantization and packet forwarding method of transmission and analyze its performance, in terms of the total network load versus the achieved distortion level (quality of service).
In section~\ref{sec:QNC}, we describe a simplified version of QNC, called one-step QNC \cite{naba4}, where QNC and packet forwarding are combined.
This simplification helps us to provide analytic results on the sufficient number of quantized network coded packets to ensure a given allowable distortion, as presented in section~\ref{sec:QNCanalysis}.
Finally, in section~\ref{sec:conclusion}, we present our conclusion remarks and discussions.

\section{Network Model and Notation}
\label{sec:ProbDesc}
To model the network, we use a directed graph, $\mathcal{G}=(\mathcal{V},\mathcal{E})$, where $\mathcal{V}=\{1,\cdots,n\}$, and $\mathcal{E}=\{1,\cdots,|\mathcal{E}|\}$, are the sets of nodes and edges, respectively.
The nodes denote the sensing and communication devices, and the edges represent communication links between pairs of nodes.
Specifically, each edge, $e$, can maintain a lossless communication from its tail node, $tail(e)$, to its head node, $head(e)$, at a maximum rate of $C_0$ bits per channel use.
Because of the lossless nature of the links, the input and output contents of edges are the same and denoted by $Y_e(t)$, at time index $t$. 
This time index represents transmission of a packet of length $L$ over edges.
We also assume that the initial rest condition holds in the network, implying: 
\begin{equation}
Y_e(1)=0,~\forall e \in \mathcal{E}.
\end{equation}
Further, the sets of incoming and outgoing edges of node $v$, are also defined as follows:
\begin{eqnarray}
{In}(v)&=&\{e \in \mathcal{E}:head(e)=v\}, \\
{Out}(v)&=&\{e \in \mathcal{E}:tail(e)=v\}.
\end{eqnarray}

The edges are uniformly distributed between pairs of nodes, forming a Erdos-Renyi random graph \cite{erdds1959random}. 
Explicitly, for each edge $e \in \mathcal{E}$, and for all $v,v' \in \mathcal{V}$, where $v \neq v'$, we have:
\begin{equation}\label{Eq:defPconn}
\textbf{P}\Big( tail(e)=v,head(e)=v' \Big) =\frac{1}{n(n-1)}.
\end{equation}

Associated to each node $v$, there is a random variable $X_v \in \mathbb{R}$, which is considered as the corresponding message of that node.
These random messages, $\underline{X}=[X_v:v \in \mathcal{V}]$, are correlated and their correlation is modeled by using sparsity.
Specifically, we assume that there is an orthonormal transform matrix, $\phi_{n \times n}$, for which $\underline{S}=\phi^T \underline{X}$ is $k$-sparse: $\vectornormZero{\underline{S}} = k$.
Moreover, we assume that the messages, take their values randomly and \textit{uniformly} between $-q_{\rm{max}}$ and $+q_{\rm{max}}$:
\begin{equation}
|X_v| \leq q_{\rm{max}},~\forall v \in \mathcal{V}.
\end{equation}
The messages, $X_v$, are ready for transmission at $t=1$.
In this paper, the lower case letters correspond to the realizations of random variables, denoted by upper case letters.

In our incast scenario, referred to as \textit{data gathering}, we need to transmit messages, $X_v$'s, to a single gateway (decoder) node, denoted by $v_0 \in \mathcal{V}$ and recover them with a maximum distortion of $D_0$, such that:
\begin{equation}\label{Eq:distortionLimit}
\textbf{E}[|X_v-\hat{X}_v|] \leq D_0,~\forall v \in \mathcal{V}.
\end{equation}
In (\ref{Eq:distortionLimit}), $\hat{X}_v$ is the recovered version of $X_v$, at the decoder node.
The product of the required number of packets to ensure the above distortion constraint and the packet length is called the \textit{total network load}, in this paper.
This total network load can be used to reflect the required number of transmissions and can be used as a measure of efficiency for different transmission methods.


Performing noiseless \cite{slepian1973noiseless} or noisy \cite{wZcoding} distributed source coding is the usual solution to take care of inter-node redundancy of messages. 
However, in the described cases, where the knowledge of appropriate marginal rates for performing separate encoding is not available at the encoder side, only a non-adaptive scheme, which does not rely on the prior of messages, can be adopted.
In this paper, we study such scenario and discuss on the feasibility of employing efficient non-adaptive transmission schemes for sparse messages.

\section{Quantization and Packet Forwarding}
\label{sec:Qforward}

In this section, we consider Packet Forwarding (PF) as a non-adaptive transmission method for the messages. 
As it was discussed earlier, distributed source coding can not be done, as the appropriate marginal rates are not known at the sensor (encoder) nodes.
The messages are going to be forwarded to the decoder node, according to the calculated routes.

Limited capacity of the links requires us to quantize the messages before transmission over the links.
This has motivated the use of phrase Quantization and Packet Forwarding (QPF), for this method.
Specifically, we quantize the messages and send the quantized version, $\textbf{Q}(X_v)$, to the decoder node.
Since the messages are uniformly distributed, a uniform quantizer minimizes the associated quantization noise power, for which:
\begin{equation}
|X_v-\textbf{Q}(X_v)| \leq \Delta_Q = \frac{2q_{\rm{max}}}{\lfloor 2^{L C_0} \rfloor}.
\end{equation}
Moreover, for the uniform quantizer, we have:
\begin{equation}\label{Eq:unQerrVar}
\textbf{E}[|X_v-\textbf{Q}(X_v)|]=\frac{\Delta_Q}{2}.
\end{equation}
This is summarized in the following corollary.

\begin{corollary}\label{th:qPFrd}
For the described QPF scenario with real-valued uniform messages, the distortion level of $D_0$,
\begin{equation}\label{Eq:distMeasure1}
\textbf{E}[|X_v-\hat{X}_v|] \leq D_0,~\forall v \in \mathcal{V},
\end{equation}
is achieved if and only if the adopted packet length, $L$, to transmit $(n-1)$ quantized messages is such that:
\begin{equation}
L \simeq \frac{1}{C_0} \log_{2}(\frac{q_{\rm{max}}}{D_0}), 
\end{equation}
resulting a total network load of
\begin{equation}\label{Eq:loadPF}
L \cdot (n-1) = \frac{n-1}{C_0} \log_{2}(\frac{q_{\rm{max}}}{D_0}).~~\blacksquare
\end{equation}
\end{corollary}

It can be understood from this corollary that the total network load is in the order of the number of nodes, $n$.
In section~\ref{sec:QNCanalysis}, we will see that by using our previously proposed one-step QNC, one requires a smaller load to ensure the same distortion level.

\section{One-Step Quantized Network Coding}
\label{sec:QNC}

Our formulation of QNC was originally presented in \cite{naba1,naba2}, where the sparse recovery was the only criterion used to recover the messages at the decoder node. 
In \cite{naba3}, we investigated the possibility of implementing optimal mean square-error decoder, when the prior information of messages is also known, in addition to their sparsity.
One-step QNC (originally mentioned in \cite{naba4}) is a special case of QNC, in which we perform linear network coding only at the first time instance and then simply forward the quantized network coded packets to the decoder node.
In the following, we explicitly describe our one-step QNC scenario and its explicit formulation.

To initiate one-step QNC, each node transmits the quantized messages, $\textbf{Q}(X_v)$'s, to its neighboring nodes. As a result of initial rest condition, at $t=2$, we have:
\begin{eqnarray}
Y_e(2) &=&\textbf{Q}(X_{tail(e)}) \\
&=& X_{tail(e)}+N_e(2), \label{Eq:nnn1}
\end{eqnarray}
where $N_e(2)$ is the quantization noise.
These quantized messages are used in each node $v$ to calculate a random linear combination, $P_v$:
\begin{equation}\label{Eq:defPv}
P_v=\sum_{e' \in In(v)} \beta_{v,e'} Y_{e'}(2)+ \alpha_{v} X_v,~\forall v \in \mathcal{V},
\end{equation}
where the local network coding coefficients, $\beta_{v,e'}$'s and $\alpha_{v}$'s, are uniformly and randomly selected from $\{-\kappa,+\kappa\}$, $\kappa>0$.
The appropriate value of the constant $\kappa$ was shown in \cite{naba4} to be equal to:
\begin{equation}
\kappa=\sqrt{\frac{2n^2}{n+|\mathcal{E}|}}.
\end{equation}
Moreover, to avoid over flow, we ensure that the normalization condition of Eq.~3 in \cite{naba1} holds, implying:
\begin{equation}\label{Eq:normalization}
\sum_{e' \in In(v)} |\beta_{v,e'}|+|\alpha_{v}| \leq 1,~\forall v \in \mathcal{V}.
\end{equation}

Next, each node randomly and independently decides to forward each $P_v$ to the decoder node or not, according to a fixed binary probability model. 
Representing the $i$'th received packet at the decoder node by $\{\underline{Z}_{\rm{tot}}(t)\}_i$, we have:
\begin{equation}
\{\underline{Z}_{\rm{tot}}(t)\}_{i}=P_v+N_e(3),~v \longrightarrow i,
\end{equation}
where $v \longrightarrow i$ means that $P_v$ is forwarded to decoder and corresponds to the $i$'th received packet.
The outgoing edge $e$, $e \in Out(v)$, is also the one on which $\textbf{Q}(P_v)$ is sent out and therefore $N_e(3)$ is the corresponding quantization noise.
The total number of received packets at the decoder node, by time instance $t$, is represented by $m$.

Linearity of this scenario lets us reformulate $\underline{Z}_{\rm{tot}}(t)$ as:
\begin{equation}\label{Eq:measEq}
\underline{Z}_{\rm{tot}}(t)=\Psi_{\rm{tot}}(t) \cdot \underline{X}+\underline{N}_{\rm{eff,tot}}(t),
\end{equation}
where $\Psi_{\rm{tot}}(t)$ and $\underline{N}_{\rm{eff,tot}}(t)$ are the \textit{total measurement matrix} and the \textit{total effective noise vector}, respectively:
\begin{equation}\label{Eq:defPsiTot}
\{\Psi_{\rm{tot}}(t)\}_{i,v}=
\left\{
\begin{array}{l l}
  \beta_{v',e'}  & ,~\scriptsize v' \rightarrow i, v \overset{e'}{\rightarrow} v', \\
  \alpha_{v'}  & ,~\scriptsize v' \rightarrow i, v'=v, \\
  0  &  ,~\mbox{otherwise} \\ \end{array} \right.
\end{equation}
\begin{equation}
\{\underline{N}_{\rm{eff,tot}}(t)\}_{i}=N_e(3)+\sum_{e' \in In(v)}\beta_{v,e'} ~N_{e'}(2),~v \rightarrow i. \label{Eq:defNeffTot}
\end{equation}
In (\ref{Eq:defPsiTot}), $v \overset{e}{\rightarrow} v'$ denotes that there is an edge $e$ from $v$ to $v'$.

Motivated by the theory of compressed sensing and sparse recovery \cite{CS,candes2007sparsity}, we claimed that one can recover messages from smaller number of received packets, $\{\underline{Z}_{\rm{tot}}(t)\}_i$'s, than the number of messages, \textit{i.e.} $m < n$ \cite{naba4}.
The decoder design and implementation for achieving near-optimal (mean squared error) performance is discussed in \cite{naba3}.
In the next section, we investigate the performance of one-step QNC, by analyzing its total network load for a given allowable distortion level, $D_0$.

\section{Distortion Analysis for One-Step Quantized Network Coding}
\label{sec:QNCanalysis}

The precise distortion analysis for QNC scenario can not be done easily because of its reliance on sparse recovery algorithms, for which a comprehensive statistical distortion analysis is still not available.
However, we are still able to obtain some performance bounds on the delay-quality performance, thanks to the work on Bayesian compressed sensing \cite{baron2010bayesian}.
In the following, we present a series of theoretical results, which characterizes the delay-distortion performance of one-step QNC scenario for non-adaptive encoding of sparse messages.

\begin{theorem}\label{th:lInf}
For the one-step QNC scenario, described in section~\ref{sec:QNC}, where $\underline{X}$ is $k$-sparse in the transform domain $\phi_{n \times n}$, with
\begin{equation}
q'_{\rm{max}}=\max_{\underline{X}} \vectornormInf{\phi^T \cdot \underline{X}},
\end{equation}
then for any $\epsilon,\gamma>0$, using a number of packets, $m$, where
\begin{equation}\label{Eq:cond1}
m > 48(1+\gamma)  \frac{(\kappa^2-1) ~ k  q'^2_{\rm{max}} + \Delta_Q^2 }{\epsilon^2} \log(n),
\end{equation}
one can decode $\underline{Z}_{\rm{tot}}(t)$ into $\hat{\underline{X}}$, such that:
\begin{equation}
\textbf{P}( |X_v-\hat{X}_v| < \epsilon ) \geq 1- n^{-\gamma},~\forall v \in \mathcal{V}.
\end{equation}
\end{theorem}
\begin{proof}
See Appendix.
\end{proof}

\begin{theorem}\label{th:finalQNC1}
For the described one-step QNC scenario, if the total network load, $m \cdot L$, satisfies the condition of Eq.~\ref{Eq:finalQNCm}, then we have:
\begin{equation}\label{Eq:disConstr1}
\textbf{E}[|X_v-\hat{X}_v|] \leq D_0 ,~\forall v \in \mathcal{V}.
\end{equation}
\end{theorem}
\begin{figure*}
\begin{eqnarray}
m \cdot L &>&  \min_{\epsilon,\gamma,L'}  48(1+\gamma)  \frac{(\kappa^2-1) ~ k  q'^2_{\rm{max}} + q_{\rm{max}}^2 2^{2-2 L' C_0} }{  \epsilon^2} L' \cdot  \log(n) \label{Eq:finalQNCm} \\
&& \mbox{subject to:}~~ \epsilon (1-n^{-\gamma}) + 2 q_{\rm{max}} n^{-\gamma} \leq D_0. \nonumber \\
&&\nonumber \\
m \cdot L & > & 96 \frac{(\kappa^2-1)k q'^2_{\rm{max}}+q^2_{\rm{max}} 2^{2-2 C_0}}{(D_0-2q_{\rm{max}}/n)^2}  \cdot \log(n) \label{Eq:pffff}
\end{eqnarray}
\end{figure*}
\begin{proof}
From theorem~\ref{th:lInf}, we know that if for a choice of $\epsilon,\gamma>0$ the condition of (\ref{Eq:cond1}) is satisfied, then the expected value of $|X_v-\hat{X}_v|$ is smaller than $\epsilon$, with probability $1-n^{-\gamma}$. 
We also know that $X_v$'s are bounded between $-q_{\rm{max}}$ and $+q_{\rm{max}}$. 
Therefore, if we limit (clip) the outcome of the decoder, $\hat{X}_v$, to be between $-q_{\rm{max}}$ and $+q_{\rm{max}}$, we can ensure that $|X_v-\hat{X}_v| \leq 2q_{\rm{max}}$, for the portion of the times ($n^{-\gamma}$) that theorem~\ref{th:lInf} fails to provide any guarantee.
Hence:
\begin{equation}
\textbf{E}[|X_v-\hat{X}_v|] \leq \epsilon (1-n^{-\gamma}) + 2q_{\rm{max}} n^{-\gamma}.
\end{equation}
And, by constraining:
\begin{equation}
\epsilon (1-n^{-\gamma}) + 2q_{\rm{max}} n^{-\gamma} \leq D_0,
\end{equation}
we can also ensure that $\textbf{E}[|X_v-\hat{X}_v|] \leq D_0$ holds. 
Finally, we perform a minimization to find the lowest possible total network load, subject to the constraint (\ref{Eq:disConstr1}).
\end{proof}

As it can be understood from (\ref{Eq:finalQNCm}), increasing the packet length, $L$, helps to decrease the required number of packets, $m$. 
On the other hand, it may increase the total network load. Theorem~\ref{th:finalQNC1} formalizes this trade-off and characterizes the appropriate choice of packet length.

By choosing $\gamma=1$, we obtain:
\begin{equation}
\epsilon \leq \frac{n D_0-2q_{\rm{max}}}{n-1}.
\end{equation}
And since
\begin{equation}
D_0-\frac{2q_{\rm{max}}}{n}=\frac{n D_0-2q_{\rm{max}}}{n} < \frac{n D_0-2q_{\rm{max}}}{n-1},
\end{equation}
by choosing $\epsilon=D_0-{2q_{\rm{max}}}/{n}$ and $L=1$, theorem~\ref{th:finalQNC1} simplifies to the following corollary.
\begin{corollary}\label{th:finalQNC2}
For the one-step QNC scenario, if the total network load, $m \cdot L$, satisfies the condition of Eq.~\ref{Eq:pffff}, then the desired distortion level of $D_0$ can be ensured:
\begin{equation}
\textbf{E}[|X_v-\hat{X}_v|] \leq D_0 ,~\forall v \in \mathcal{V}.~\blacksquare
\end{equation}
\end{corollary}

Although the preceding corollary loosen our former formulation on the conditions to ensure the distortion constraint, it provides us with a better understanding.
Explicitly, when the number of network nodes, $n$, is large, the total network load in one-step QNC scenario, $m \cdot L$, (Eq.~\ref{Eq:pffff}) is in the order of $\log(n)$.
As discussed in section~\ref{sec:Qforward}, the total network load in QPF scenario, $(n-1) \cdot L$, has order of $n$, for the same distortion level. 
This clarifies a significant reduction in the total network load in one-step QNC scenario, for large networks.

\section{Conclusion}
\label{sec:conclusion}

The possibility of non-adaptive distributed compression of correlated sources has been investigated in this paper.
This was done by discussing the performance of our previously proposed one-step QNC scenario, in which random linear network coding meets the theory of compressed sensing and sparse recovery.
Specifically, we derived a sufficient (not necessary) condition on the network load to satisfy a desired distortion level.
Our mathematical derivations show a significant decrease in the total load of the network, compared to the case of conventional QPF.
The resulting decrease in the total network load shows a potential decrease in the overall transmission power in the network. 
One may also weakly interpret this decrease as a decrease on the delivery delay, required to achieve the desired distortion level.



\section*{Appendix: Proof of theorem~\ref{th:lInf}}
It is shown in the proof of theorem~3.1 in \cite{naba4} that $\Psi_{\rm{tot}}(t)$ is such that for $\forall i,v,~1 \leq i \leq m, 1 \leq v \leq n$, $\{\Psi_{\rm{tot}}(t)\}_{i,v}$'s are independent, and we have (Eq.~20 in \cite{naba4}):
\begin{eqnarray}
\textbf{E}[\{\Psi_{\rm{tot}}(t)\}_{iv}] &=& 0, \nonumber \\
\textbf{E}[\{\Psi_{\rm{tot}}(t)\}^2_{iv}] &=& 1, \nonumber \\
\textbf{E}[\{\Psi_{\rm{tot}}(t)\}^4_{iv}] &=& \kappa^2. \label{Eq:MeasMatConds}
\end{eqnarray}

Define positive integers $m_1$ and $m_2$, which we will determine, and set $m=m_1 m_2$.
Partition the $m \times n$ matrix $\Psi_{\rm{tot}}(t)$ into $m_2$ matrices $\{\Psi_1,\ldots,\Psi_{m_2}\}$, each of size $m_1 \times n$.
Now, consider a realization of $\underline{X}$, called $\underline{x}$, for which we define $\{\underline{z}_1=\frac{1}{\sqrt{m_1}}\Psi_1 \underline{x}+\underline{n}_1,\ldots,\underline{z}_{m_2}=\frac{1}{\sqrt{m_1}}\Psi_{m_2} \underline{x}+\underline{n}_{m_2}\}$. 
Moreover, $\{\underline{n}_l\}_i$'s are the corresponding elements of the total effective noise vector, for which (\ref{Eq:defNeffTot}) implies: $|\{n_{l}\}_i| \leq  \Delta_Q$.

Furthermore, we define $\underline{u}_j$'s, where $1 \leq j \leq n$, to be the canonical vectors in $\mathbb{R}^n$.
For them, we define $\{\underline{z}'_{1,j}=\frac{1}{\sqrt{m_1}}\Psi_1 \underline{u}_j,\ldots,\underline{z}'_{m_2,j}=\frac{1}{\sqrt{m_1}}\Psi_{m_2} \underline{u}_j\}$.

Now, we define random variables $r_{1,j},\ldots,r_{{m_2},j}$:
\begin{eqnarray}
r_{l,j}=\underline{z}_l^T \cdot \underline{z}'_{l,j}&=&\sum_{i=1}^{m_1} \{\underline{z}_l\}_{i}~ \{\underline{z}'_{l,j}\}_{i}=\sum_{i=1}^{m_1} \frac{1}{m_1} w_{l,j,i}, \nonumber
\end{eqnarray}
where:
\begin{equation}
w_{l,j,i}= \Big ( \sum_{v=1}^{n} \{\Psi_{l}\}_{i,v} x_v + \{\underline{n}_l\}_{i} \Big ) \cdot \Big ( \sum_{v=1}^{n} \{\Psi_{l}\}_{i,v} \{\underline{u}_j\}_{v} \Big ).\nonumber
\end{equation}
Moreover, $r_{1,j},\ldots,r_{{m_2},j}$ are independent because of the independence of $\{\Psi_{l}\}_{i,v}$'s.

By using a similar reasoning as in \cite{wang2007distributed}, we have:
\begin{eqnarray}
\textbf{E}[w_{l,j,i}] &=& 
\textbf{E}[\sum_{v=1}^{n} \{\Psi_{l}\}^2_{i,v} ~x_v ~\{\underline{u}_j\}_{v} ] \nonumber \\
&& +2\textbf{E}[\sum_{v=1}^{n} \sum_{v'=v+1}^{n} \{\Psi_{l}\}_{i,v} \{\Psi_{l}\}_{i,v'}  x_v \{\underline{u}_j\}_{v'} ] \nonumber \\
&& +\{\underline{n}_l\}_{i}~\textbf{E}[ (\sum_{v=1}^{n} \{\Psi_{l}\}_{i,v} \{\underline{u}_j\}_{v})] \nonumber \\
&=& \sum_{j=1}^{n}  x_v ~\{\underline{u}_j\}_{v} = \underline{x}^T \cdot \underline{u}_j, \nonumber 
\end{eqnarray}
\begin{eqnarray}
\textbf{E}[w^2_{l,j,i}] &=& 
\textbf{E}[(\sum_{v=1}^{n} \{\Psi_{l}\}_{i,v} x_v)^2 \cdot (\sum_{v=1}^{n} \{\Psi_{l}\}_{i,v} \{\underline{u}_j\}_{v})^2] \nonumber \\
&& +\{\underline{n}_l\}_{i}^2 ~\textbf{E}[ (\sum_{v=1}^{n} \{\Psi_{l}\}_{i,v} \{\underline{u}_j\}_{v})^2] \nonumber \\
&& +2 \{\underline{n}_l\}_{i} ~\textbf{E}[ (\sum_{v=1}^{n} \{\Psi_{l}\}_{i,v} x_v) (\sum_{v=1}^{n} \{\Psi_{l}\}_{i,v} \{\underline{u}_j\}_{v})^2] \nonumber \\
&=& 2(\underline{x}^T \cdot \underline{u}_j)^2+\vectornorm{\underline{x}}^2 ~ \vectornorm{\underline{u}_j}^2 \nonumber \\
&& +(\kappa^2-3) (\sum_{v=1} x_v^2 \{\underline{u}_l\}^2_{v} ) + \{\underline{n}_l\}_{i}^2 \vectornorm{\underline{u}_j}^2,  \nonumber
\end{eqnarray}
where in the last equality we used Lemma~1 in \cite{wang2007distributed}, and also the fact that $\textbf{E}[\{\Psi_{l}\}^3_{i,v}]=0$. 
Recalling that $\underline{u}_j$'s are canonical vectors with unit $\ell_2$-norm, we have:
\begin{eqnarray}
\textbf{Var}[w_{l,j,i}] &=&\vectornorm{\underline{x}}^2 +(\kappa^2-2) x^2_j + \{\underline{n}_l\}_{i}^2 . \nonumber
\end{eqnarray}

Now, for $r_{l,j}$, we have:
\begin{eqnarray}
\textbf{E}[r_{l,j}]& =& \sum_{i=1}^{m_1} \frac{1}{m_1} \textbf{E}[w_{l,i}]= \underline{x}^T \cdot \underline{u}_j, \nonumber 
\end{eqnarray}
\begin{eqnarray}
\textbf{Var}[r_{l,j}] &=& \frac{1}{m_1^2} \sum_{i=1}^{m_1} \textbf{Var}[w_{l,j,i}] \nonumber \\
&=& \frac{1}{m_1} \Big ( \vectornorm{\underline{x}}^2 +(\kappa^2-2) x^2_j + \{\underline{n}_l\}_{i}^2 \Big )~. \nonumber
\end{eqnarray}
Thus, the Chebyshev inequality,
\begin{eqnarray}
\textbf{P}(|r_{l,j}-\underline{x}^T \underline{u}| \geq \epsilon ) 
&\leq & \frac{\textbf{Var}(r_{l,j})}{\epsilon^2 }, \nonumber
\end{eqnarray}
and the fact that $|\{\underline{n}_l\}_{i}| \leq \Delta_Q$, imply:
\begin{eqnarray}
\textbf{P}(|r_{l,j}-x_j| \geq \epsilon ) 
& \leq & \frac{1}{\epsilon^2 m_1} \Big ( \vectornorm{\underline{x}}^2 +(\kappa^2-2) x^2_j + \Delta_Q^2   \Big ) \nonumber  \\
& \leq & \frac{1}{\epsilon^2 m_1} \Big ( (\kappa^2-1) \max_{\underline{x}} \vectornorm{\underline{x}}^2 
+ \Delta_Q^2   \Big ).   \nonumber
\end{eqnarray}
By picking
\begin{equation}\label{Eq:m1}
m_1 > 4 \frac{(\kappa^2-1) \max_{\underline{x}} \vectornorm{\underline{x}}^2 + \Delta_Q^2 }{\epsilon^2},
\end{equation}
we have:
\begin{equation}
\textbf{P}(|r_{l,j}-x_j|>\epsilon) < 
\frac{1}{4}.
\end{equation}

We define $\hat{x}_j$, the decoded value for $x_j$, to be the median of the independent random variables $r_{1,j},\ldots,r_{m_2,j}$.
Further, we define the indicator random variable $\xi_{l,j}$ to be equal to one if $|r_{l,j}-x_j|>\epsilon$ and zero otherwise.
Now, if $\sum_{l=1}^{m_2} \xi_{l,j}$ is more than $\frac{m_2}{2}$, the median (decoded value) will not be inside the interval: $|\hat{x}_j-x_j|>\epsilon$. Moreover, we have: $\textbf{E}[\sum_{l=1}^{m_2} \xi_{l,j}] < \frac{m_2}{4}$.
By using the Chernoff inequality
\begin{equation}
\textbf{P}(\sum_{l=1}^{m_2} \xi_{l,j} > (1+\delta)\textbf{E}[\sum_{l=1}^{m_2} \xi_{l,j}]) \leq e^{-\delta^2 \textbf{E}[\sum_{l=1}^{m_2} \xi_{l,j}]/3 }
\end{equation}
where $0<\delta<1$ \cite{papoulis2001probability}, we obtain:
\begin{eqnarray}\label{Eq:ll1}
\textbf{P}(\sum_{l=1}^{m_2} \xi_{l,j}>(1+\delta)\frac{m_2}{4}  ) <  e^{-\delta^2 m_2/12}. 
\end{eqnarray}
Taking the limit of (\ref{Eq:ll1}) when $\delta \rightarrow 1^{-}$,
\begin{eqnarray}
\textbf{P}(\sum_{l=1}^{m_2} \xi_{l,j}>\frac{m_2}{2}  ) <  e^{-m_2/12}. 
\end{eqnarray}
Finally, by using the union bound, we have:
\begin{equation}
\textbf{P}\Big( \exists j: |x_j-\hat{x}_j|>\epsilon \Big ) \leq n \cdot e^{-m_2/12}.
\end{equation}
Therefore, by choosing 
\begin{equation}
m_2 > 12(1+\gamma) \log(n),
\end{equation}
and combining with (\ref{Eq:m1}) to find $m=m_1 \cdot m_2$, we have proved the theorem.

\bibliographystyle{IEEEtran}
\bibliography{Ref_arXiv}

\end{document}